\documentclass[a4paper,10pt,oneside]{article}
\usepackage[top=2.50cm, bottom=0.5cm,left=2.5cm, right=1.5cm]{geometry}
\usepackage{amssymb,amsmath,amsthm}
\usepackage[numbers]{natbib}
\usepackage{amsmath}
\date{}
\usepackage{graphicx}
\usepackage{caption}
\usepackage{subcaption}
\usepackage{tabularx}
\usepackage{multirow}
\usepackage{float}
\usepackage[colorlinks]{hyperref}
\usepackage{multicol}
\usepackage{adjustbox}
\usepackage{colortbl}
\usepackage{pdflscape}
\usepackage{algorithm}
\usepackage{algorithmic}
\usepackage{epstopdf}
\usepackage[affil-it]{authblk}
\title{\bf{Computational Study of New Record-Based Transmuted Chen Distribution and Its Applications to the Failure Time and Iron Sheet Data}}
\author{Caner Tan{\i}\c{s}}\affil{Department of Statistics, \c{C}ank{\i}r{\i} Karatekin University, \c{C}ank{\i}r{\i}, Turkey; canertanis@karatekin.edu.tr}
\newtheorem{theorem}{Theorem}

\begin{document}
\maketitle
\begin{abstract}
This study is considered to introduce a novel distribution as an alternative to Chen distribution via the record-based transmutation method. This technique is based on the distributions of first two upper record values. Thus, we suggest a new special case based on Chen distribution in the family of record-based transmuted distributions. We explore various distributional properties of the proposed model namely, quantile function, hazard function, median, moments, and stochastic ordering. Our distribution has three parameters and to estimate these parameters, we utilize nine different and well-known estimators. Then, we compare the performances of these estimators via a comprehensive simulation study. Also, we provide two real-world data examples to assess the fits the data sets the suggested model and its some competitors.
\end{abstract}

\section{Introduction}
The modeling of lifetime and reliability data is a foundation  in various scientific and engineering disciplines, necessitating flexible statistical distributions able of modeling various hazard rate behaviors. Classic models, such as the exponential and Weibull distributions, often fall short in capturing complex hazard functions, including non-monotonic structures like bathtub or upside-down bathtub shapes. To address these limitations, researchers have developed generalized and extended distributions that offer greater adaptability.

The researchers suggest new methods to generate the flexible lifetime distribution. One of these techniques is record-based transmutation introduced by \cite{balakrishnan2021record}. The record-based transmutation map (RBTM) is based on distributions of the first two upper record values. The RBTM is summarized as follows. \newline
Let $X_1 $ and $X_2 $  be a random sample with two sizes from the distribution with cumulative distribution function (CDF) $G(.)$ and the probability density function (PDF) $g(.)$, and $X_{U\left( 1 \right)} $ and $X_{U\left( 2 \right)} $ 
be upper records associated with the sample.
 
 Let us define a random variable $Y$ 
\[
\begin{array}{l}
 Y\overset{d}{=} \mbox{ }X_{U\left( 1 \right)} ,\mbox{ 
with probability }p _1, \\ 
 Y\overset{d}{=} \mbox{ }X_{U\left( 2 \right)} ,\mbox{ 
with probability }p _2, \\ 
 \end{array}
\]
where $U_{\left( n\right) }=\min \left\{ i:i>U\left( n-1\right)
,X_{i}>X_{U\left( n-1\right) }\right\} \left\{ U_{\left( n\right) }\right\}
_{n=1}^{\infty }$ refer to upper record times and $\left\{ X_{U\left(
n\right) }\right\} _{n=1}^{\infty }$ denotes the corresponding record sequence \cite{balakrishnan2021record},\citet{arnold2008first}, $p_{1}+p_{2}=1$ and $Y$ refers to a random variable having  record-based transmuted distribution. In this regard, the CDF of $Y$ is 
\begin{eqnarray}
F_{Y}\left( x\right)  &=&p_{1}P\left( {X_{U\left( 1\right) }\leq x}\right)
+p_{2}P\left( {X_{U\left( 2\right) }\leq x}\right)   \nonumber \\
&=&G\left( x\right) +p\left[ {\left( {1-G\left( x\right) }\right) \log
\left( {1-G\left( x\right) }\right) }\right], 
\end{eqnarray}
where $p\in \left( {0,1} \right)$. The corresponding PDF is  
 
\begin{equation}
\label{eq9}
f_Y \left( x \right)=g\left( x \right)\left[ {1+p\left( {-\log \left( 
{1-G\left( x \right)} \right)-1} \right)} \right].
\end{equation}

The RBTM has been applied to various baseline distributions, namely Weibull \cite{tanics2022record}, power Lomax \cite{sakthivel2022record}, Lindley \cite{tanics2024new}, generalized linear exponential \cite{arshad2024record}, unit Omega \cite{pathak2024record}, Burr X \cite{alrweili2025statistical}.

The main purpose of this paper is to introduce a novel submodel of the record-based transmuted family based on Chen distribution to provide higher flexibility in modeling lifetime data, providing a wider range of hazard rate behaviors. 

The rest of this study is designed as follows. Section 2 introduces a new submodel of record-based transmuted family of distribution related to Chen distribution. In Section 3, some basic statistical properties are explored such as density shapes, quantile function, hazard function, median, moments, and stochastic ordering. Then, we use nine methods to estimate the parameters of the proposed model in Section 4. To compare these estimators, we consider a Monte Carlo simulation study in Section 5. Section 6 provides two real-world data examples to compare the fits of the suggested model and its competitors. Finally, the conclusions are given in Section 7.

\section{Record-based transmuted Chen distribution}
Chen \cite{chen2000new} suggested Chen distribution with cumulative distribution function (CDF) and probability density function (PDF) are 
\begin{align}
\label{CDF CHEN}
G(x)=1 - \exp\left(\omega \left(1 - \exp(x^\kappa)\right)\right),
\end{align}
and
\begin{align}
\label{PDF CHEN}
g(x)=\omega \kappa x^{\kappa - 1} \exp \left( \omega \left( 1 - \exp \left( x^\kappa \right) \right) + x^\kappa \right); x>0,~\omega,~\kappa>0.
\end{align}

By using a transformation defined by \cite{balakrishnan2021record} on the CDF \eqref{CDF CHEN}, we obtain the record-based transmuted Chen (RBTC) distribution, with the CDF and PDF given as follows:
\begin{align}
\label{CDF RBTC}
F(x; \omega, \kappa,p)=1 - \exp\left(\omega \left(1 - \exp(x^\kappa)\right)\right) + p \omega \left(1 - \exp(x^\kappa)\right) \exp\left(\omega \left(1 - \exp(x^\kappa)\right)\right)
\end{align}

and
\begin{align}
\label{PDF RBTC}
f(x; \omega, \kappa,p)=\omega \kappa x^{\kappa-1} \exp\left(\omega \left(1 - \exp(x^\kappa)\right) + x^\kappa\right) \left(1 - p\omega \left(1 - \exp(x^\kappa)\right) - p\right).
\end{align}
where $~x>0,~\omega,~\kappa,~0<p<1$.
Figure \ref{pdf} illustrates the possible shapes of PDFs for RBTC distribution. From Figure \ref{pdf} it is clear that the density can be shaped as increasing, decreasing and unimodal for the chosen parameters.

\begin{figure} [H]
    \centering
    \includegraphics[width=0.8\linewidth]{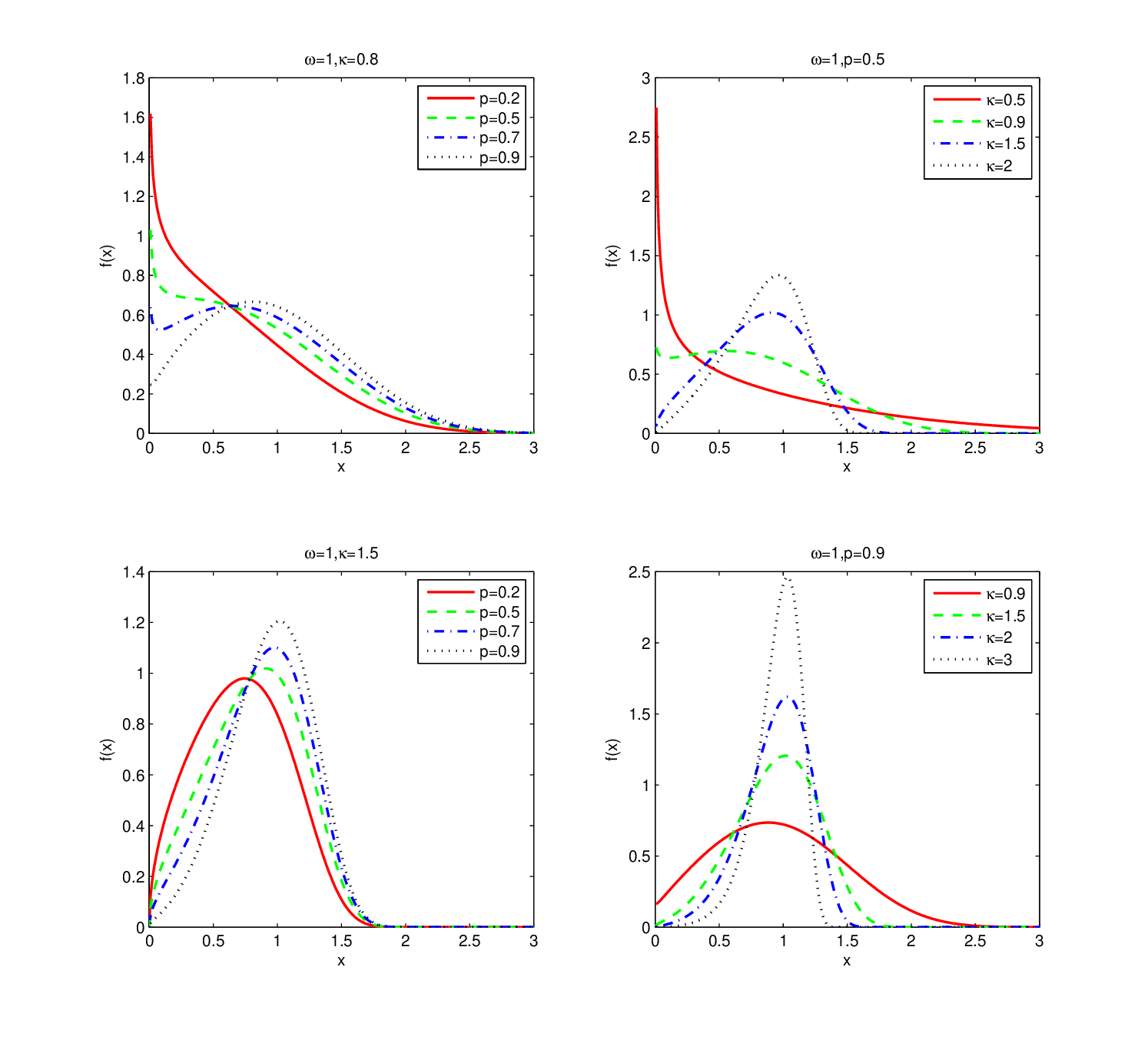}
    \caption{The PDFs of RBTC distribution for the selected parameter values}
\label{pdf}
\end{figure}

\section{Some properties of the RBTC distribution}
\label{properties}
In this section, we examine some distributional properties of the RBTC distribution, namely, quantile function, hazard function (hf), median, moments, and stochastic ordering.

\subsection{Quantile function}
 The quantile function of RBTC distribution is given as follows:
\begin{align}
\label{QF RBTC}
 Q(x) = \exp\left( \frac{ \log \left( \log \left( -\frac{1 - p \omega + p \, \mathrm{LambertW}\left( \frac{-1+u}{p e^{1/p}} \right)}{\omega p} \right) \right) }{ \kappa } \right)
\end{align}
where, $0<u<1$ and \( W(\cdot) \) denotes the Lambert W function is defined as the inverse relation of the function \( w \mapsto w e^w \), that is: $W(z) \cdot e^{W(z)} = z$.

 \subsection{Hazard function}
The hf of the RBTC distribution is 
\begin{equation}
h(x)=\omega \kappa x^{\kappa-1} \frac{ e^{x^\kappa} \left( (p\omega + p - 1) - p\omega e^{x^\kappa} \right) }{p\omega(1-e^{x^\kappa})-1 }
\end{equation}

Figure \ref{hf} illustrates the possible shapes of HFs for RBTC distribution. From Figure \ref{hf} it is clear that the hf can be shaped as increasing, decreasing and bathtub curve for the selected parameters.

\begin{figure} [H]
    \centering
    \includegraphics[width=0.8\linewidth]{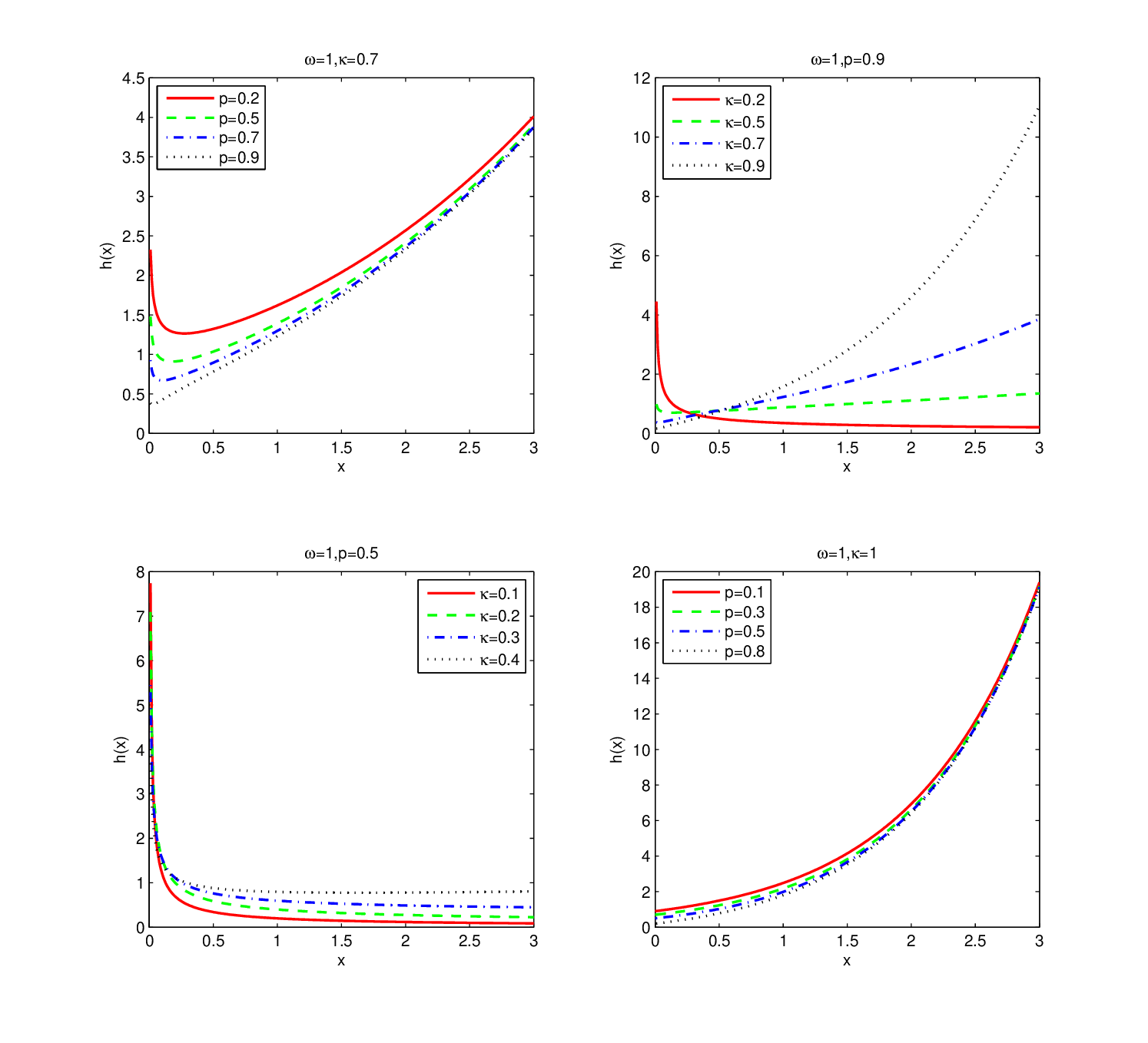}
    \caption{The HFs of RBTC distribution for the selected parameter values}
\label{hf}
\end{figure}

\subsection{Median}
The median of the RBTC distribution can be obtained by substituting 0.5 for $u$ in Eq. (\ref{QF RBTC}). Thus, the corresponding median is
\begin{align}
\label{Med RBTEP}
\mathrm{Median}(X) = \exp\left( \frac{ \log \left( \log \left( -\frac{1 - p \omega + p \, \mathrm{LambertW}\left( \frac{-1/2}{p e^{1/p}} \right)}{\omega p} \right) \right) }{ \kappa } \right).
\end{align}


\subsection{Moments}
In this  subsection, we introduce the $r$-th central moment of $RBTC(\omega,\kappa,p)$ distribution via the following theorem.
\begin{theorem}
    
Let $X$ be a random variable from the $RBTC(\omega,\kappa,p)$ distribution. The $r$-th central moment is

\begin{equation}
\label{moment}
\mathbb{E}[X^r] = e^{\omega} \sum_{i=0}^\infty \binom{r-\kappa}{i} 
(-\log \omega)^{r-\kappa - i} 
\left[ (1 - p \omega - p) J(i, 0, \omega) + p \, J(i, 1, \omega) \right],
\end{equation}
where \[
J(i, m, \omega) := \int_\omega^\infty (\log t)^i e^{-t} t^m \, dt, \quad m \in \{0,1\}.
\]
\end{theorem}
\begin{proof}

To compute the integral in Eq. (\ref{moment}), we substitute $exp(x^\kappa)=u$ and rewrite the integral as follows.

The \(r\)-th moment $\mathbb{E}[X^r]$ is defined as
\[
\mathbb{E}[X^r] = (1 - p \omega - p) \, \omega e^{\omega} I_1 + p \, \omega^{2} e^{\omega} I_2,
\]
where
\[
I_1 = \int_1^\infty (\log u)^{r - \kappa} e^{-\omega u} \, du, \quad 
I_2 = \int_1^\infty u (\log u)^{r - \kappa} e^{-\omega u} \, du.
\]

Then, we expand the logarithmic term using the generalized binomial theorem:
\[
(\log u)^{r - \kappa} = \bigl(-\log \omega + \log(\omega u)\bigr)^{r - \kappa} = 
\sum_{i=0}^\infty \binom{r-\kappa}{i} (-\log \omega)^{r-\kappa - i} \left(\log(\omega u)\right)^i,
\]

Substitute the expansion into \(I_1\) and \(I_2\), and apply the transform of variable \(t = \omega u\). Thus, we derive
\[
I_1 = \frac{1}{\omega} \sum_{i=0}^\infty \binom{r-\kappa}{i} 
(-\log \omega)^{r-\kappa - i} \int_\omega^\infty (\log t)^i e^{-t} t^{0} \, dt,
\]
\[
I_2 = \frac{1}{\omega^{2}} \sum_{i=0}^\infty \binom{r-\kappa}{i} 
(-\log \omega)^{r-\kappa - i} \int_\omega^\infty (\log t)^i e^{-t} t^{1} \, dt.
\]

Finally, we get
\[
\mathbb{E}[X^r]=e^{\omega} \sum_{i=0}^\infty \binom{r-\kappa}{i} 
(-\log \omega)^{r-\kappa - i} 
\left[ (1 - p \omega - p) J(i, 0, \omega) + p \, J(i, 1, \omega) \right].
\]
Thus, the proof is completed.
\end{proof}

\subsection{Stochastic Ordering}
In this subsection, we explore the stochastic ordering for the RBTC distribution. Therefore, we define the following theorem.
\begin{theorem}
\label{lrordering} Let $X\sim RBTC\left( \omega ,\kappa ,p_{1}\right) $ and
$Y\sim RBTC\left( \omega ,\kappa ,p_{2}\right) .$ If $p_{1}<p_{2}$ then $X$
is less than $Y$ in the likelihood ratio order, that is, the ratio function
of the corresponding PDFs decreases in $x$.
\end{theorem}

\begin{proof}
The ratio of densities for any $x>0$ is as follows:
\begin{equation*}
g\left( x\right) =.\frac{1 - p_1 \omega + p_1 \omega e^{x^{\kappa}} - p_1}{1 - p_2 \omega + p_2 \omega e^{x^{\kappa}} - p_2}
\end{equation*}

Then, consider the derivative of $\log \left( g\left( x\right) \right) $ in $x$%
\begin{equation*}
\frac{d\log \left( g\left( x\right) \right) }{dx}=\kappa x^{\kappa - 1} e^{x^\kappa} \left( \frac{p_1 \omega}{1 - p_1 (1 + \omega) + p_1 \omega e^{x^\kappa}} - \frac{p_2 \omega}{1 - p_2 (1 + \omega) + p_2 \omega e^{x^\kappa}} \right)<0
\end{equation*}%
for $p_{1}<p_{2}$, and thus the proof is completed.
\end{proof}

We notice that  it follows that $X$ is also less than $Y$ in the hazard ratio, the mean residual life, and the stochastic orders under the conditions
given in Theorem \ref{lrordering} from \cite{shaked1994stochastic}.

\section{Parameter Estimation}\label{Estimation methods}
In this section, we discuss point estimation for the RBTC distribution. To estimate the parameters of RBTC distribution, we utilize nine different methods. 

\subsection{Method of maximum likelihood}
In this subsection, we propose ML estimators (MLEs) of the parameters $\omega$, $\kappa$, and $p$ of the RBTC distribution. 

Let ${X_1},{X_2}...{X_n}$ be a random sample from the RBTC and ${x_1},{x_2}...{x_n}$  denote the observed values of the sample. Then, the corresponding log-likelihood function is given by

\begin{align}
\label{mle}
\ell(\omega, \kappa, p) = \sum_{i=1}^n \left[
\log \omega + \log \kappa + (\kappa - 1) \log x_i + \omega \left(1 - e^{x_i^\kappa} \right) + x_i^\kappa + \log\left(1 - p\omega \left(1 - e^{x_i^\kappa} \right) - p\right)
\right]
\end{align}

We obtain the MLEs of the parameters $\omega, \kappa$ and $p$ of the RBTC by maximizing Eq. (\ref{mle}). This optimization problem is solved via numerical methods such as Newton-Raphson and Nelder Mead. 
\subsection{Method of least squares}
This subsection provides the LS estimators (LSEs) the $\omega, \kappa$ and $p$ parameters of the RBTC distribution. This estimator is proposed by \cite{swain1988least} as an alternative to the MLE. The LSEs can be derived by minimizing the function given in Eq. (\ref{lse}).

\begin{align}
\label{lse}
LS(x_{i})&=\sum_{i=1}^{n}\left[F(x_{i:n})-\frac{i}{n+1}\right]^2,
\end{align}
where ${x_{i:n}}$ for $i = 1,2...n$ refer to the order statistics.
\subsection{Method of weighted least squares}
This subsection introduces the WLS estimators (WLSEs) of $\omega, \kappa$ and $p$ parameters for the RBTC via the method proposed by \cite{swain1988least}. We derive the WLSEs by minimizing the Eq. (\ref{wlse}).
\begin{align}
\label{wlse}
WLS(x_{i})&=\sum_{i=1}^{n}\frac{(n+1)^2(n+2)}{i(n-i+1)}\left[F(x_{i:n})-\frac{i}{n+1}\right]^2.
\end{align} 

\subsection{Method of Anderson-Darling}
In this subsection, we discuss the AD estimators (ADEs) of $\omega, \kappa$ and $p$ parameters. This method is related to the AD goodness-of-fit statistic proposed by \cite{anderson1952asymptotic}. The ADEs ${\hat \omega _{ADE}},{\hat \kappa _{ADE}}$, ${\hat p _{ADE}}$ of the parameters $\phi, \kappa$ and $p$ can be obtained by minimizing Eq (\ref{ade}).

\begin{align}
\label{ade}
AD(x_{i})&=-n-\frac{1}{n}\sum_{i=1}^{n}(2i-1)\left[\log F(x_{i:n})+\log S(x_{n-i-1:n})\right].
\end{align} 


\subsection{Method of Cram\'{e}r-von Mises}
In this section, we provide the CvM estimators (CvMEs) of parameters of RBTC distribution. The Cram\'{e}r-von Mises method depends on minimizing the difference between the CDF and the empirical distribution functions, as proposed by \cite{choi1968estimation}. The CvMEs are computed by maximizing the function in Eq. (\ref{cvme}).

\begin{align}
\label{cvme}
C(x_{i})&=\frac{1}{12n}+\sum_{i=1}^{n}\left[F(x_{i:n})-\frac{2i-1}{2n}\right]^2.
\end{align}
\subsection{Method of maximum product of spacings}
In this subsection, we deal with the point estimation for the RBTC distribution via the MPS method. The MPS method introduced by \cite{cheng1983estimating} and \cite{ranneby1984maximum} as an alternative to the ML method, the MPS method based on maximizing the function in Eq. (\ref{mpse}) to derive the MPS estimators (MPSEs). 
\begin{align}
\label{mpse}
\delta\left(x_{i} \right) =\frac{1}{n+1}\sum_{i=1}^{n+1}\log I_{i}(x_{i}),
\end{align}

where $I_{i}(x_{i})=F(x_{i:n})-F(x_{i-1:n})$, $F(x_{0:n})=0$ and $F(x_{n+1:n})=1.$

\subsection{Right tail Anderson Darling estimation method}
This subsection proposes the RTAD estimators (RTADEs) of $\omega$, $\kappa$, and $p$ parameters. The RTADEs are obtained by minimizing the Eq. (\ref{RTADE}).
\begin{equation} \label{RTADE}
    L\left( {\omega,\kappa,p} \right) = \frac{n}{2} - 2\sum\limits_{i = 1}^n {F\left( {{x_{i:n}}} \right) - \frac{1}{n}} \sum\limits_{i = 1}^n {\left( {2i - 1} \right)\log \left( {1 - F\left( {{x_{n - i + 1}}} \right)} \right)} 
\end{equation}

\subsection{Minimum spacing absolute distance estimation method}
In this subsection, it is proposed the MSAD estimators (MSADEs) of the parameters $\omega$, $\kappa$, and $p$. The MSADEs of $\omega$, $\kappa$, and $p$ parameters are obtained by minimizing the function given in Eq. (\ref{msad}).
\begin{equation} \label{msad}
    \Lambda \left( {\omega,\kappa,p} \right|{x_i}) = \sum\limits_{i = 1}^{n + 1} {\left| {{D _i} - \frac{1}{{n + 1}}} \right|} .
\end{equation}
\subsection{Minimum spacing absolute-log distance estimation method}
In this subsection, we discuss the MSALD estimators (MSALDEs) of $\omega$, $\kappa$, and $p$ parameters. \cite{torabi2008general} suggested the MSALD method as an alternative the ML. We obtain the MSALDEs of the $\phi$, $\kappa$, and $p$ parameters by minimizing the Eq. (\ref{MSALD})
\begin{equation}\label{MSALD}
    \Psi \left( {\omega,\kappa,p|{x_i}} \right) = \sum\limits_{i = 1}^{n + 1} {\left| {{D_i} - \log \frac{1}{{n + 1}}} \right|}. 
\end{equation}

\section{Simulation} \label{simulation}
In this section, we provide an extensive Monte-Carlo (MC) Simulation study to evaluate the performance of the mentioned estimators in Section  \ref{Estimation methods}. We design the parameter settings and other issues in the MC simulations as follows:
In all MC simulations, the sample sizes, $n=25,50,100,200,500,1000$ with 5000 repetitions and the initial values of the $\omega$, $\kappa$ and $p$ parameters are considered as follows: \newline
$Case_{I} = \left( {\omega  = 2,\kappa  = 1,p = 0.5} \right)$,\\
$Case_{II} = \left( {\omega  = 0.9,\kappa  = 0.9,p = 0.9} \right)$,\\
$Case_{III} = \left( {\omega  = 1.5,\kappa  = 0.5,p = 0.3} \right)$,\\
$Case_{IV} = \left( {\omega  = 2.2,\kappa  = 0.7,p = 0.2} \right)$,\\

We compare the performances of the mentioned estimators via the bias, mean squared error (MSE), and mean relative error (MRE) values. The corresponding formulas of these measures are 

$$Bias=\frac{1}{5000}\sum\limits_{i=1}^{5000}\left( \hat{\Phi}-\Phi
\right),$$

$$MSE=\frac{1}{5000}\sum\limits_{i=1}^{5000}\left( \hat{\Phi%
}-\Phi \right) ^{2},$$

$$MRE=\frac{1}{5000}\sum_{i=1}^{5000}\frac{|\hat{\Phi}-\Phi|}{\Phi},$$
where $\Phi=\left(\omega,\kappa,p\right)$. 
\subsection{Random sample generation}
In this subsection, we explore the generation of random samples from the $RBTC(\omega, \kappa, p)$ distribution. In this regard, we define an acceptance-rejection (AR) sampling algorithm.
We prefer the Weibull distribution as the proposal distribution in the AR algorithm since it is a popular distribution. The AR algorithm is presented as follows:

\textbf{Algorithm 1.}

\textbf{A1.} Generate data on random variable $Y \sim Weibull(\vartheta,\varpi)$ 
with the PDF $g$ given as follows: 
\begin{equation*}
g\left( \vartheta,\varpi \right)=\vartheta \varpi x^{\varpi -1}e^{-\vartheta
x^\varpi}.
\end{equation*}

\textbf{A2.} Generate $U$ from standard uniform distribution(independent of $%
Y$).

\textbf{A3.} If%
\begin{equation*}
U<\frac{f\left( Y;\omega,\kappa,p\right) }{k\times g\left(
Y;\vartheta,\varpi \right) }
\end{equation*}%
then set $X=Y$ (\textquotedblleft accept\textquotedblright ); otherwise go
back to A1 (\textquotedblleft reject\textquotedblright ), where the PDF $f$ $%
\left( .\right) $ is given as in Eq. (\ref{PDF RBTC})  and 
\begin{equation*}
k=\underset{z\in 
\mathbb{R}
_{+}}{\max }\frac{f\left( z;\omega,\kappa,p\right) }{g\left(
z;\vartheta,\varpi \right) }.
\end{equation*}%

We generate random data on $X$ from the $RBTC(\omega, \kappa, p)$ via AR algorithm. We use Algorithm 1 in all MC simulations.

The results of MC simulations are given in Tables \ref{sim:tab1}-\ref{sim:tab4}.

\begin{table}[H]
\centering
\caption{The biases, MSEs and MREs for $\omega=2$, $\kappa=1$ and $p=0.5$}
\scalebox{0.9} {\ \label{sim:tab1}
\begin{tabular}{ccccccccccc}\hline
          &      &         & Bias    &         &        & MSE    &        &        & MRE    &        \\\hline
Estimator & $n$    & $\hat\omega$   & $\hat\kappa$   & $\hat p$  & $\hat\omega$  & $\hat\kappa$  & $\hat p$ & $\hat\omega$  & $\hat\kappa$  & $\hat p$ \\\hline

	MLE       & 25   & 0.01076  & 0.13366  & -0.15703 & 0.21178 & 0.07558 & 0.09668 & 0.18534 & 0.21125 & 0.52461 \\
          & 50   & -0.06968 & 0.07679  & -0.13849 & 0.14140 & 0.03769 & 0.09564 & 0.15394 & 0.15485 & 0.52292 \\
          & 100  & -0.06788 & 0.03499  & -0.08245 & 0.12446 & 0.02361 & 0.09041 & 0.14416 & 0.12471 & 0.50545 \\
          & 200  & -0.04724 & 0.01240  & -0.04522 & 0.08796 & 0.01697 & 0.07899 & 0.12054 & 0.10699 & 0.46448 \\
          & 500  & -0.05257 & 0.00819  & -0.04252 & 0.07038 & 0.01010 & 0.06073 & 0.10332 & 0.08336 & 0.39718 \\
          & 1000 & -0.04354 & 0.00633  & -0.03325 & 0.04870 & 0.00700 & 0.04211 & 0.08435 & 0.06905 & 0.32169 \\\hline
LSE       & 25   & -0.02841 & -0.04684 & 0.02119  & 0.17550 & 0.06082 & 0.12776 & 0.16788 & 0.19900 & 0.50768 \\
          & 50   & -0.06814 & -0.02918 & -0.02282 & 0.13530 & 0.03856 & 0.09813 & 0.14875 & 0.15641 & 0.48225 \\
          & 100  & -0.06909 & -0.01408 & -0.02778 & 0.10717 & 0.02700 & 0.09196 & 0.13338 & 0.12820 & 0.47484 \\
          & 200  & -0.03138 & -0.02714 & 0.01040  & 0.08161 & 0.01874 & 0.07883 & 0.11622 & 0.10605 & 0.44076 \\
          & 500  & -0.04708 & -0.01711 & -0.00937 & 0.07140 & 0.01148 & 0.06411 & 0.10482 & 0.08414 & 0.39831 \\
          & 1000 & -0.02195 & -0.01915 & 0.00933  & 0.04478 & 0.00936 & 0.04389 & 0.08168 & 0.07538 & 0.32750 \\\hline
WLSE      & 25   & -0.01232 & -0.05475 & 0.01925  & 0.21988 & 0.05702 & 0.19188 & 0.18340 & 0.19311 & 0.57686 \\
          & 50   & -0.10043 & 0.00002  & -0.08017 & 0.16272 & 0.03334 & 0.10716 & 0.16278 & 0.14529 & 0.50765 \\
          & 100  & -0.07206 & -0.00037 & -0.04840 & 0.11704 & 0.02346 & 0.08841 & 0.13946 & 0.12116 & 0.47401 \\
          & 200  & -0.04200 & -0.01254 & -0.01563 & 0.08314 & 0.01580 & 0.07607 & 0.11577 & 0.10072 & 0.43834 \\
          & 500  & -0.05333 & -0.00498 & -0.02841 & 0.07007 & 0.00990 & 0.06061 & 0.10158 & 0.08022 & 0.38658 \\
          & 1000 & -0.03205 & -0.00698 & -0.01174 & 0.04533 & 0.00743 & 0.04078 & 0.08114 & 0.06810 & 0.31201 \\\hline
ADE       & 25   & 0.04686  & -0.01604 & 0.02448  & 0.18287 & 0.05540 & 0.11489 & 0.17142 & 0.18644 & 0.51918 \\
          & 50   & -0.03589 & -0.00846 & -0.03144 & 0.12583 & 0.03217 & 0.09635 & 0.14525 & 0.14463 & 0.49941 \\
          & 100  & -0.02882 & -0.01833 & -0.00555 & 0.12124 & 0.02155 & 0.08721 & 0.14289 & 0.11750 & 0.47782 \\
          & 200  & -0.01211 & -0.02559 & 0.01584  & 0.07788 & 0.01665 & 0.07601 & 0.11340 & 0.10312 & 0.44246 \\
          & 500  & -0.04345 & -0.00941 & -0.01825 & 0.07088 & 0.01000 & 0.06043 & 0.10265 & 0.08030 & 0.38756 \\
          & 1000 & -0.01551 & -0.01417 & 0.00556  & 0.04109 & 0.00781 & 0.03897 & 0.07860 & 0.06964 & 0.30863 \\\hline
CvME      & 25   & 0.04933  & 0.05713  & -0.05014 & 0.22959 & 0.07806 & 0.13565 & 0.18862 & 0.21425 & 0.54419 \\
          & 50   & -0.05837 & 0.03121  & -0.08276 & 0.14717 & 0.04115 & 0.10482 & 0.15548 & 0.15879 & 0.50843 \\
          & 100  & -0.06872 & 0.01717  & -0.06055 & 0.11339 & 0.02844 & 0.09765 & 0.13778 & 0.13106 & 0.49605 \\
          & 200  & -0.04514 & -0.00581 & -0.02044 & 0.08836 & 0.01843 & 0.08422 & 0.12189 & 0.10631 & 0.45845 \\
          & 500  & -0.04600 & -0.01152 & -0.01444 & 0.07214 & 0.01158 & 0.06541 & 0.10599 & 0.08493 & 0.40376 \\
          & 1000 & -0.02579 & -0.01437 & 0.00219  & 0.04669 & 0.00940 & 0.04534 & 0.08347 & 0.07583 & 0.33332 \\\hline
MPSE      & 25   & -0.02870 & -0.09641 & 0.05519  & 0.14251 & 0.05805 & 0.10715 & 0.14784 & 0.19692 & 0.50954 \\
          & 50   & -0.01265 & -0.08764 & 0.06059  & 0.10521 & 0.03700 & 0.09048 & 0.12964 & 0.15991 & 0.50240 \\
          & 100  & -0.00050 & -0.07652 & 0.06703  & 0.10174 & 0.02619 & 0.08797 & 0.12480 & 0.13793 & 0.50075 \\
          & 200  & 0.02576  & -0.06634 & 0.07594  & 0.07406 & 0.01976 & 0.07602 & 0.10827 & 0.11918 & 0.46772 \\
          & 500  & 0.01336  & -0.04284 & 0.04668  & 0.05967 & 0.01163 & 0.05731 & 0.09519 & 0.09060 & 0.39565 \\
          & 1000 & 0.01802  & -0.03225 & 0.04105  & 0.03828 & 0.00763 & 0.03758 & 0.07623 & 0.07262 & 0.31469 \\\hline
TADE      & 25   & 0.09375  & 0.03309  & 0.02844  & 0.20998 & 0.07052 & 0.16209 & 0.18625 & 0.20992 & 0.62016 \\
          & 50   & 0.07561  & -0.01152 & 0.06042  & 0.16424 & 0.04214 & 0.11303 & 0.16371 & 0.16392 & 0.54516 \\
          & 100  & 0.03034  & -0.01783 & 0.04672  & 0.10457 & 0.02875 & 0.10163 & 0.13504 & 0.13803 & 0.53100 \\
          & 200  & 0.03730  & -0.02517 & 0.05345  & 0.07928 & 0.01882 & 0.07805 & 0.11731 & 0.11191 & 0.47140 \\
          & 500  & -0.00457 & -0.00940 & 0.01037  & 0.05744 & 0.01142 & 0.05364 & 0.09845 & 0.08965 & 0.39419 \\
          & 1000 & -0.01362 & -0.00472 & 0.00022  & 0.04186 & 0.00834 & 0.03890 & 0.08414 & 0.07766 & 0.33357 \\\hline
MSADE     & 25   & 0.03354  & -0.05738 & 0.04603  & 0.10800 & 0.05697 & 0.06688 & 0.11971 & 0.19128 & 0.35642 \\
          & 50   & 0.01837  & -0.05732 & 0.04935  & 0.07629 & 0.03770 & 0.05831 & 0.10392 & 0.15483 & 0.36639 \\
          & 100  & 0.02751  & -0.04433 & 0.05611  & 0.07049 & 0.02643 & 0.05439 & 0.10059 & 0.13230 & 0.36247 \\
          & 200  & 0.03879  & -0.04549 & 0.05649  & 0.05247 & 0.01800 & 0.04932 & 0.08755 & 0.10866 & 0.34526 \\
          & 500  & 0.02251  & -0.02983 & 0.03823  & 0.04109 & 0.01042 & 0.04124 & 0.07720 & 0.08290 & 0.31698 \\
          & 1000 & 0.01809  & -0.02684 & 0.03471  & 0.03181 & 0.00735 & 0.03116 & 0.06685 & 0.06913 & 0.26704 \\\hline
MSALDE    & 25   & 0.05785  & -0.09512 & 0.08038  & 0.15664 & 0.06202 & 0.09480 & 0.14986 & 0.20337 & 0.47624 \\
          & 50   & 0.02833  & -0.08656 & 0.07437  & 0.11593 & 0.04019 & 0.09002 & 0.13438 & 0.16523 & 0.49021 \\
          & 100  & 0.01821  & -0.07030 & 0.07674  & 0.09618 & 0.02912 & 0.08496 & 0.12282 & 0.14280 & 0.49541 \\
          & 200  & 0.04374  & -0.06706 & 0.08330  & 0.08118 & 0.02158 & 0.07797 & 0.11372 & 0.12295 & 0.46725 \\
          & 500  & 0.02398  & -0.04462 & 0.05489  & 0.05827 & 0.01281 & 0.05736 & 0.09517 & 0.09498 & 0.39811 \\
          & 1000 & 0.01581  & -0.03311 & 0.04052  & 0.04532 & 0.00871 & 0.04223 & 0.08290 & 0.07722 & 0.33396

\\\hline
\end{tabular}
}
\end{table}

\begin{table}[H]
\centering
\caption{The biases, MSEs and MREs for $\omega=0.9$, $\kappa=0.9$ and $p=0.9$}
\scalebox{0.9} {\ \label{sim:tab2}
\begin{tabular}{ccccccccccc}\hline
          &      &         & Bias    &         &        & MSE    &        &        & MRE    &        \\\hline
Estimator & $n$    & $\hat\omega$   & $\hat\kappa$   & $\hat p$  & $\hat\omega$  & $\hat\kappa$  & $\hat p$ & $\hat\omega$  & $\hat\kappa$  & $\hat p$ \\\hline

MLE       & 25   & -0.21750 & 0.17698  & -0.41008 & 0.09731 & 0.06498 & 0.24896 & 0.29668 & 0.22601 & 0.46250 \\
          & 50   & -0.16074 & 0.11728  & -0.30654 & 0.07565 & 0.03911 & 0.19261 & 0.24345 & 0.16934 & 0.35905 \\
          & 100  & -0.09787 & 0.06275  & -0.19457 & 0.05076 & 0.02137 & 0.12759 & 0.17814 & 0.11681 & 0.24769 \\
          & 200  & -0.03054 & 0.01497  & -0.08323 & 0.02243 & 0.00982 & 0.05455 & 0.10802 & 0.07401 & 0.13806 \\
          & 500  & 0.00958  & -0.01018 & -0.01512 & 0.00578 & 0.00267 & 0.01075 & 0.05889 & 0.04100 & 0.06487 \\
          & 1000 & 0.02321  & -0.02031 & 0.00684  & 0.00231 & 0.00126 & 0.00224 & 0.04396 & 0.03292 & 0.04187 \\\hline
LSE       & 25   & -0.10365 & 0.04571  & -0.22808 & 0.07620 & 0.04401 & 0.14502 & 0.26240 & 0.17829 & 0.35878 \\
          & 50   & -0.06317 & 0.02878  & -0.15618 & 0.06075 & 0.02902 & 0.12495 & 0.22574 & 0.15326 & 0.30755 \\
          & 100  & -0.03086 & 0.01007  & -0.09220 & 0.04453 & 0.02072 & 0.09028 & 0.18385 & 0.12963 & 0.23826 \\
          & 200  & 0.01444  & -0.02164 & -0.01838 & 0.02408 & 0.01263 & 0.04736 & 0.13080 & 0.10029 & 0.16337 \\
          & 500  & 0.04268  & -0.04071 & 0.03017  & 0.00879 & 0.00559 & 0.01331 & 0.08370 & 0.06813 & 0.09653 \\
          & 1000 & 0.05425  & -0.05103 & 0.04852  & 0.00573 & 0.00406 & 0.00637 & 0.07091 & 0.06125 & 0.07389 \\\hline
WLSE      & 25   & -0.14381 & 0.06912  & -0.29470 & 0.06992 & 0.03594 & 0.16862 & 0.24872 & 0.15694 & 0.37656 \\
          & 50   & -0.09463 & 0.05041  & -0.20540 & 0.05628 & 0.02555 & 0.12890 & 0.21186 & 0.14129 & 0.29959 \\
          & 100  & -0.05128 & 0.02463  & -0.12261 & 0.03840 & 0.01646 & 0.08126 & 0.16263 & 0.10930 & 0.20952 \\
          & 200  & -0.00533 & -0.00584 & -0.04644 & 0.01838 & 0.00875 & 0.03768 & 0.10903 & 0.07914 & 0.13210 \\
          & 500  & 0.02323  & -0.02417 & 0.00221  & 0.00626 & 0.00355 & 0.01050 & 0.06576 & 0.05111 & 0.07131 \\
          & 1000 & 0.03540  & -0.03382 & 0.02165  & 0.00320 & 0.00209 & 0.00287 & 0.05214 & 0.04316 & 0.04906 \\\hline
ADE       & 25   & -0.13010 & 0.08242  & -0.27739 & 0.06573 & 0.03427 & 0.14845 & 0.23772 & 0.15788 & 0.32664 \\
          & 50   & -0.08873 & 0.05656  & -0.19520 & 0.04905 & 0.02298 & 0.10557 & 0.19299 & 0.13127 & 0.25501 \\
          & 100  & -0.05080 & 0.02781  & -0.12072 & 0.03412 & 0.01443 & 0.07240 & 0.15022 & 0.10090 & 0.18721 \\
          & 200  & -0.00451 & -0.00501 & -0.04435 & 0.01603 & 0.00753 & 0.03188 & 0.10170 & 0.07400 & 0.11885 \\
          & 500  & 0.02331  & -0.02397 & 0.00255  & 0.00574 & 0.00316 & 0.00914 & 0.06352 & 0.04947 & 0.06747 \\
          & 1000 & 0.03558  & -0.03405 & 0.02198  & 0.00315 & 0.00210 & 0.00276 & 0.05196 & 0.04329 & 0.04847 \\\hline
CvME      & 25   & -0.09431 & 0.09501  & -0.21296 & 0.08811 & 0.06228 & 0.16292 & 0.28234 & 0.21352 & 0.37831 \\
          & 50   & -0.05521 & 0.04882  & -0.14354 & 0.06538 & 0.03448 & 0.13331 & 0.23458 & 0.16621 & 0.31610 \\
          & 100  & -0.02600 & 0.01868  & -0.08441 & 0.04672 & 0.02246 & 0.09488 & 0.18898 & 0.13369 & 0.24416 \\
          & 200  & 0.01792  & -0.01832 & -0.01282 & 0.02489 & 0.01294 & 0.04888 & 0.13303 & 0.10049 & 0.16614 \\
          & 500  & 0.04440  & -0.03968 & 0.03296  & 0.00896 & 0.00553 & 0.01354 & 0.08469 & 0.06760 & 0.09773 \\
          & 1000 & 0.05513  & -0.05055 & 0.04995  & 0.00582 & 0.00402 & 0.00651 & 0.07160 & 0.06082 & 0.07480 \\\hline
MPSE      & 25   & -0.13672 & 0.02510  & -0.30714 & 0.05798 & 0.02366 & 0.15304 & 0.21642 & 0.13242 & 0.34405 \\
          & 50   & -0.10938 & 0.02888  & -0.23586 & 0.04551 & 0.01725 & 0.11574 & 0.17916 & 0.11330 & 0.26468 \\
          & 100  & -0.05499 & 0.00642  & -0.12941 & 0.02469 & 0.01004 & 0.05566 & 0.12098 & 0.08356 & 0.15513 \\
          & 200  & -0.01445 & -0.01196 & -0.06054 & 0.01203 & 0.00586 & 0.02624 & 0.08146 & 0.06243 & 0.09444 \\
          & 500  & 0.01270  & -0.02028 & -0.01270 & 0.00406 & 0.00228 & 0.00635 & 0.05202 & 0.04115 & 0.05299 \\
          & 1000 & 0.02392  & -0.02534 & 0.00629  & 0.00209 & 0.00138 & 0.00180 & 0.04193 & 0.03505 & 0.03772 \\\hline
TADE      & 25   & -0.09824 & 0.08698  & -0.20773 & 0.10148 & 0.05166 & 0.20805 & 0.30636 & 0.19744 & 0.42826 \\
          & 50   & -0.06710 & 0.05024  & -0.15923 & 0.08397 & 0.03545 & 0.17878 & 0.27326 & 0.17236 & 0.37880 \\
          & 100  & -0.04246 & 0.02376  & -0.10965 & 0.05757 & 0.02201 & 0.12570 & 0.21083 & 0.13050 & 0.28687 \\
          & 200  & -0.01357 & -0.00088 & -0.06294 & 0.03572 & 0.01335 & 0.08324 & 0.15735 & 0.09963 & 0.21564 \\
          & 500  & 0.01271  & -0.01811 & -0.01551 & 0.01373 & 0.00550 & 0.03015 & 0.09313 & 0.06312 & 0.12223 \\
          & 1000 & 0.02482  & -0.02774 & 0.00521  & 0.00590 & 0.00275 & 0.01017 & 0.06533 & 0.04686 & 0.07736 \\\hline
MSADE     & 25   & -0.06584 & 0.00675  & -0.20867 & 0.04599 & 0.03168 & 0.09393 & 0.18876 & 0.14666 & 0.28071 \\
          & 50   & -0.05771 & 0.01223  & -0.15289 & 0.03323 & 0.01863 & 0.06434 & 0.15509 & 0.11796 & 0.21355 \\
          & 100  & -0.03443 & -0.00012 & -0.09774 & 0.02181 & 0.01110 & 0.03997 & 0.12245 & 0.09075 & 0.15270 \\
          & 200  & -0.00648 & -0.01342 & -0.04954 & 0.01216 & 0.00640 & 0.01954 & 0.09153 & 0.06937 & 0.10027 \\
          & 500  & 0.00831  & -0.01670 & -0.01673 & 0.00587 & 0.00293 & 0.00926 & 0.06218 & 0.04651 & 0.06772 \\
          & 1000 & 0.02015  & -0.02156 & 0.00133  & 0.00313 & 0.00180 & 0.00340 & 0.04929 & 0.03845 & 0.04880 \\\hline
MSALDE    & 25   & -0.09898 & 0.01345  & -0.27582 & 0.05344 & 0.02905 & 0.11787 & 0.20571 & 0.14778 & 0.31302 \\
          & 50   & -0.09028 & 0.02386  & -0.21414 & 0.04181 & 0.02075 & 0.09464 & 0.17944 & 0.12621 & 0.24534 \\
          & 100  & -0.04887 & 0.00562  & -0.12165 & 0.02543 & 0.01158 & 0.04914 & 0.13142 & 0.09255 & 0.15310 \\
          & 200  & -0.01545 & -0.01092 & -0.06394 & 0.01447 & 0.00638 & 0.02753 & 0.09486 & 0.06865 & 0.10665 \\
          & 500  & 0.00763  & -0.01700 & -0.02000 & 0.00650 & 0.00312 & 0.01187 & 0.06382 & 0.04741 & 0.06735 \\
          & 1000 & 0.02230  & -0.02329 & 0.00406  & 0.00283 & 0.00167 & 0.00282 & 0.04834 & 0.03724 & 0.04604

\\\hline
\end{tabular}
}
\end{table}

\begin{table}[H]
\centering
\caption{The biases, MSEs and MREs for $\omega=1.5$, $\kappa=0.5$ and $p=0.3$}
\scalebox{0.9} {\ \label{sim:tab3}
\begin{tabular}{ccccccccccc}\hline
          &      &         & Bias    &         &        & MSE    &        &        & MRE    &        \\\hline
Estimator & $n$    & $\hat\omega$   & $\hat\kappa$   & $\hat p$  & $\hat\omega$  & $\hat\kappa$  & $\hat p$ & $\hat\omega$  & $\hat\kappa$  & $\hat p$ \\\hline

MLE       & 25   & 0.03545  & 0.03842  & -0.04275 & 0.11999 & 0.01149 & 0.06413 & 0.18618 & 0.16032 & 0.72978 \\
          & 50   & -0.00175 & 0.02009  & -0.03916 & 0.09550 & 0.00616 & 0.06511 & 0.16707 & 0.12576 & 0.74031 \\
          & 100  & 0.00264  & 0.00883  & -0.01525 & 0.08151 & 0.00358 & 0.06512 & 0.16007 & 0.09647 & 0.72696 \\
          & 200  & -0.03227 & 0.00418  & -0.03467 & 0.07412 & 0.00233 & 0.06417 & 0.15698 & 0.07783 & 0.73385 \\
          & 500  & -0.06103 & 0.00604  & -0.05687 & 0.06443 & 0.00139 & 0.05532 & 0.14690 & 0.06082 & 0.67854 \\
          & 1000 & -0.05619 & 0.00360  & -0.04815 & 0.05253 & 0.00107 & 0.04613 & 0.12807 & 0.05258 & 0.60008 \\\hline
LSE       & 25   & -0.00362 & -0.01998 & 0.02357  & 0.10408 & 0.01435 & 0.09230 & 0.17747 & 0.18818 & 0.87061 \\
          & 50   & -0.02715 & -0.01158 & -0.00993 & 0.09163 & 0.00745 & 0.07134 & 0.16285 & 0.13800 & 0.80272 \\
          & 100  & -0.02051 & -0.00797 & -0.00640 & 0.07936 & 0.00421 & 0.06410 & 0.15850 & 0.10226 & 0.76017 \\
          & 200  & -0.04497 & -0.00535 & -0.03108 & 0.06986 & 0.00239 & 0.06004 & 0.15174 & 0.07797 & 0.73121 \\
          & 500  & -0.04893 & -0.00209 & -0.03598 & 0.06277 & 0.00157 & 0.05620 & 0.14674 & 0.06465 & 0.70726 \\
          & 1000 & -0.04718 & -0.00418 & -0.03000 & 0.05435 & 0.00117 & 0.04975 & 0.13384 & 0.05670 & 0.65225 \\\hline
WLSE      & 25   & 0.00870  & -0.02170 & 0.00962  & 0.13933 & 0.01206 & 0.18823 & 0.19638 & 0.17640 & 1.02838 \\
          & 50   & 0.05864  & -0.02518 & 0.07668  & 0.10309 & 0.00773 & 0.10977 & 0.17193 & 0.13876 & 0.95067 \\
          & 100  & 0.03517  & -0.01522 & 0.04671  & 0.08416 & 0.00425 & 0.08280 & 0.16345 & 0.10193 & 0.83338 \\
          & 200  & -0.02704 & -0.00718 & -0.01576 & 0.08108 & 0.00270 & 0.07433 & 0.16280 & 0.08209 & 0.79643 \\
          & 500  & -0.03080 & -0.00299 & -0.02135 & 0.06356 & 0.00151 & 0.05854 & 0.14557 & 0.06299 & 0.69991 \\
          & 1000 & -0.05299 & -0.00109 & -0.03981 & 0.05479 & 0.00109 & 0.04920 & 0.13248 & 0.05368 & 0.63282 \\\hline
ADE       & 25   & 0.02042  & -0.00530 & 0.01257  & 0.10826 & 0.00970 & 0.09023 & 0.17949 & 0.15384 & 0.87949 \\
          & 50   & -0.00817 & -0.00368 & -0.00893 & 0.09206 & 0.00645 & 0.07678 & 0.16498 & 0.12791 & 0.83265 \\
          & 100  & 0.00899  & -0.00846 & 0.01359  & 0.08566 & 0.00339 & 0.07512 & 0.16443 & 0.09386 & 0.80626 \\
          & 200  & -0.03799 & -0.00298 & -0.02996 & 0.07235 & 0.00229 & 0.06377 & 0.15689 & 0.07795 & 0.75530 \\
          & 500  & -0.04270 & -0.00068 & -0.03388 & 0.06331 & 0.00151 & 0.05703 & 0.14626 & 0.06389 & 0.69815 \\
          & 1000 & -0.05447 & -0.00036 & -0.04178 & 0.05288 & 0.00107 & 0.04743 & 0.13120 & 0.05369 & 0.62755 \\\hline
CvME      & 25   & 0.03738  & 0.02584  & -0.03237 & 0.12824 & 0.01735 & 0.08456 & 0.19257 & 0.19424 & 0.83166 \\
          & 50   & -0.01165 & 0.01182  & -0.04117 & 0.09591 & 0.00793 & 0.06606 & 0.16472 & 0.13846 & 0.77330 \\
          & 100  & -0.02416 & 0.00546  & -0.03261 & 0.07817 & 0.00416 & 0.06044 & 0.15666 & 0.09993 & 0.73360 \\
          & 200  & -0.04414 & 0.00058  & -0.04065 & 0.06900 & 0.00240 & 0.05954 & 0.15095 & 0.07731 & 0.72649 \\
          & 500  & -0.05103 & 0.00081  & -0.04240 & 0.06209 & 0.00155 & 0.05545 & 0.14593 & 0.06420 & 0.70165 \\
          & 1000 & -0.05093 & -0.00226 & -0.03598 & 0.05419 & 0.00114 & 0.04919 & 0.13358 & 0.05605 & 0.64735 \\\hline
MPSE      & 25   & 0.03067  & -0.05701 & 0.11171  & 0.10759 & 0.01186 & 0.12873 & 0.18013 & 0.18092 & 1.00597 \\
          & 50   & 0.06301  & -0.04829 & 0.11805  & 0.11392 & 0.00823 & 0.11592 & 0.18829 & 0.14923 & 1.00229 \\
          & 100  & 0.07905  & -0.03851 & 0.11573  & 0.09958 & 0.00566 & 0.10634 & 0.18203 & 0.12277 & 0.95393 \\
          & 200  & 0.05568  & -0.03126 & 0.08594  & 0.09675 & 0.00420 & 0.10038 & 0.18380 & 0.10321 & 0.93211 \\
          & 500  & 0.01679  & -0.01692 & 0.03584  & 0.08146 & 0.00235 & 0.07738 & 0.16860 & 0.07662 & 0.81765 \\
          & 1000 & -0.00187 & -0.01072 & 0.01362  & 0.06142 & 0.00156 & 0.05745 & 0.14069 & 0.06125 & 0.68222 \\\hline
TADE      & 25   & 0.05590  & 0.00225  & 0.04523  & 0.11454 & 0.01274 & 0.12337 & 0.18323 & 0.17470 & 0.87421 \\
          & 50   & 0.01532  & -0.00054 & 0.01288  & 0.08418 & 0.00639 & 0.08894 & 0.15916 & 0.12675 & 0.79213 \\
          & 100  & 0.00878  & -0.00150 & 0.00976  & 0.07540 & 0.00369 & 0.07422 & 0.15327 & 0.09612 & 0.73495 \\
          & 200  & -0.00618 & -0.00596 & 0.00239  & 0.06782 & 0.00235 & 0.06929 & 0.14672 & 0.07377 & 0.71230 \\
          & 500  & -0.02009 & -0.00216 & -0.01281 & 0.05742 & 0.00149 & 0.05504 & 0.13552 & 0.05978 & 0.64943 \\
          & 1000 & -0.03473 & -0.00214 & -0.02300 & 0.05038 & 0.00107 & 0.04761 & 0.12573 & 0.05148 & 0.60539 \\\hline
MSADE     & 25   & 0.00511  & -0.02472 & 0.03987  & 0.04370 & 0.01418 & 0.03935 & 0.10199 & 0.19036 & 0.49838 \\
          & 50   & 0.02592  & -0.02484 & 0.05503  & 0.04081 & 0.00721 & 0.03821 & 0.09603 & 0.13627 & 0.48650 \\
          & 100  & 0.02296  & -0.01535 & 0.04144  & 0.03401 & 0.00444 & 0.03322 & 0.08651 & 0.10475 & 0.44144 \\
          & 200  & 0.02663  & -0.01381 & 0.03966  & 0.03547 & 0.00263 & 0.03660 & 0.08725 & 0.08010 & 0.44359 \\
          & 500  & 0.01173  & -0.00507 & 0.01727  & 0.03153 & 0.00144 & 0.03138 & 0.08122 & 0.05777 & 0.40389 \\
          & 1000 & -0.01123 & -0.00334 & -0.00231 & 0.03069 & 0.00092 & 0.02819 & 0.08229 & 0.04628 & 0.39777 \\\hline
MSALDE    & 25   & 0.06197  & -0.04976 & 0.11742  & 0.11033 & 0.01311 & 0.11115 & 0.17945 & 0.18764 & 0.92318 \\
          & 50   & 0.08084  & -0.04342 & 0.11908  & 0.11715 & 0.00861 & 0.11078 & 0.18938 & 0.15102 & 0.95058 \\
          & 100  & 0.08734  & -0.03429 & 0.11412  & 0.10219 & 0.00598 & 0.10152 & 0.17915 & 0.12596 & 0.90542 \\
          & 200  & 0.05693  & -0.02829 & 0.08141  & 0.09872 & 0.00424 & 0.09884 & 0.18094 & 0.10215 & 0.90265 \\
          & 500  & 0.02802  & -0.01806 & 0.04789  & 0.08450 & 0.00269 & 0.08256 & 0.16852 & 0.08132 & 0.82907 \\
          & 1000 & -0.00296 & -0.01020 & 0.01266  & 0.06338 & 0.00165 & 0.05921 & 0.14338 & 0.06305 & 0.69000

\\\hline
\end{tabular}
}
\end{table}

\begin{table}[H]
\centering
\caption{The biases, MSEs and MREs for $\omega=2.2$, $\kappa=0.7$ and $p=0.2$}
\scalebox{0.9} {\ \label{sim:tab4}
\begin{tabular}{ccccccccccc}\hline
          &      &         & Bias    &         &        & MSE    &        &        & MRE    &        \\\hline
Estimator & $n$    & $\hat\omega$   & $\hat\kappa$   & $\hat p$  & $\hat\omega$  & $\hat\kappa$  & $\hat p$ & $\hat\omega$  & $\hat\kappa$  & $\hat p$ \\\hline

MLE       & 25   & 0.02997  & 0.01056  & 0.04733 & 0.20167 & 0.01923 & 0.06632 & 0.16228 & 0.15140 & 1.07099 \\
          & 50   & 0.08430  & -0.00895 & 0.05819 & 0.18603 & 0.01144 & 0.06622 & 0.15884 & 0.12043 & 1.07248 \\
          & 100  & 0.08588  & -0.01625 & 0.07416 & 0.15632 & 0.00663 & 0.06854 & 0.14864 & 0.09355 & 1.10131 \\
          & 200  & 0.06529  & -0.01891 & 0.06090 & 0.13115 & 0.00463 & 0.06420 & 0.13748 & 0.07526 & 1.06309 \\
          & 500  & 0.02206  & -0.01538 & 0.03389 & 0.09512 & 0.00280 & 0.04885 & 0.12082 & 0.05683 & 0.92973 \\
          & 1000 & 0.01131  & -0.01442 & 0.02700 & 0.08341 & 0.00222 & 0.04305 & 0.11263 & 0.04759 & 0.84772 \\\hline
LSE       & 25   & -0.10655 & -0.05594 & 0.09381 & 0.20883 & 0.03325 & 0.10951 & 0.16749 & 0.21157 & 1.41023 \\
          & 50   & -0.05988 & -0.04559 & 0.05047 & 0.16904 & 0.01677 & 0.07906 & 0.15375 & 0.15128 & 1.25901 \\
          & 100  & -0.01584 & -0.03190 & 0.05289 & 0.13057 & 0.00925 & 0.07037 & 0.13674 & 0.11142 & 1.19824 \\
          & 200  & 0.02024  & -0.02882 & 0.05506 & 0.13101 & 0.00606 & 0.06641 & 0.13938 & 0.09119 & 1.16822 \\
          & 500  & 0.01946  & -0.02892 & 0.05515 & 0.10577 & 0.00381 & 0.05663 & 0.12895 & 0.07129 & 1.06400 \\
          & 1000 & -0.00166 & -0.02384 & 0.03603 & 0.09278 & 0.00295 & 0.05078 & 0.12391 & 0.05955 & 0.98810 \\\hline
WLSE      & 25   & -0.07251 & -0.05315 & 0.10380 & 0.21955 & 0.02788 & 0.11631 & 0.16928 & 0.19227 & 1.44988 \\
          & 50   & -0.03074 & -0.04997 & 0.06761 & 0.21886 & 0.01640 & 0.11628 & 0.17088 & 0.14628 & 1.44992 \\
          & 100  & 0.02850  & -0.04399 & 0.09415 & 0.16806 & 0.01043 & 0.11054 & 0.15071 & 0.11623 & 1.43209 \\
          & 200  & 0.00074  & -0.02935 & 0.04533 & 0.14803 & 0.00591 & 0.08125 & 0.14696 & 0.08941 & 1.26577 \\
          & 500  & 0.00361  & -0.02440 & 0.03896 & 0.10974 & 0.00359 & 0.05901 & 0.13103 & 0.06702 & 1.06330 \\
          & 1000 & -0.01485 & -0.01895 & 0.02026 & 0.09517 & 0.00243 & 0.04838 & 0.12293 & 0.05278 & 0.94617 \\\hline
ADE       & 25   & -0.04828 & -0.03897 & 0.09281 & 0.16974 & 0.02335 & 0.10306 & 0.15115 & 0.17342 & 1.37287 \\
          & 50   & 0.02506  & -0.04149 & 0.08685 & 0.17098 & 0.01522 & 0.08979 & 0.15423 & 0.14192 & 1.30623 \\
          & 100  & 0.02552  & -0.02940 & 0.06798 & 0.13659 & 0.00789 & 0.07392 & 0.14048 & 0.10328 & 1.22092 \\
          & 200  & 0.04763  & -0.02848 & 0.06798 & 0.13293 & 0.00551 & 0.07029 & 0.14006 & 0.08617 & 1.17599 \\
          & 500  & 0.01724  & -0.02440 & 0.04642 & 0.10228 & 0.00353 & 0.05667 & 0.12777 & 0.06695 & 1.04040 \\
          & 1000 & 0.00028  & -0.01921 & 0.02891 & 0.08708 & 0.00244 & 0.04595 & 0.11862 & 0.05296 & 0.91769 \\\hline
CvME      & 25   & -0.01369 & 0.00855  & 0.01593 & 0.24965 & 0.03461 & 0.09089 & 0.17479 & 0.20500 & 1.27844 \\
          & 50   & -0.00950 & -0.01398 & 0.01355 & 0.18519 & 0.01525 & 0.06994 & 0.15933 & 0.13900 & 1.18765 \\
          & 100  & 0.00996  & -0.01530 & 0.03363 & 0.13411 & 0.00834 & 0.06373 & 0.13749 & 0.10501 & 1.14794 \\
          & 200  & 0.03248  & -0.02065 & 0.04554 & 0.13162 & 0.00568 & 0.06415 & 0.13907 & 0.08777 & 1.14570 \\
          & 500  & 0.02223  & -0.02478 & 0.04867 & 0.10451 & 0.00350 & 0.05415 & 0.12791 & 0.06848 & 1.04101 \\
          & 1000 & -0.00072 & -0.02166 & 0.03246 & 0.09175 & 0.00281 & 0.04943 & 0.12298 & 0.05822 & 0.97389 \\\hline
MPSE      & 25   & -0.12375 & -0.11657 & 0.20110 & 0.16411 & 0.03116 & 0.15307 & 0.15162 & 0.21289 & 1.64626 \\
          & 50   & 0.01416  & -0.09629 & 0.18819 & 0.14998 & 0.02194 & 0.14007 & 0.14487 & 0.17625 & 1.59909 \\
          & 100  & 0.06379  & -0.07593 & 0.17325 & 0.16160 & 0.01336 & 0.12657 & 0.15690 & 0.13452 & 1.53763 \\
          & 200  & 0.07083  & -0.05989 & 0.13903 & 0.14932 & 0.00941 & 0.11001 & 0.15360 & 0.10882 & 1.42958 \\
          & 500  & 0.05828  & -0.04281 & 0.10108 & 0.12509 & 0.00576 & 0.08371 & 0.14358 & 0.08191 & 1.23581 \\
          & 1000 & 0.03969  & -0.03145 & 0.07121 & 0.10455 & 0.00394 & 0.06461 & 0.12934 & 0.06345 & 1.05055 \\\hline
TADE      & 25   & -0.05625 & -0.02389 & 0.07531 & 0.16039 & 0.02625 & 0.13164 & 0.14893 & 0.18134 & 1.32986 \\
          & 50   & -0.02989 & -0.01948 & 0.02450 & 0.14761 & 0.01242 & 0.07151 & 0.14596 & 0.12493 & 1.12005 \\
          & 100  & 0.00596  & -0.01373 & 0.03216 & 0.12537 & 0.00650 & 0.05643 & 0.13373 & 0.08925 & 1.03661 \\
          & 200  & -0.02586 & -0.00991 & 0.00060 & 0.10081 & 0.00396 & 0.04789 & 0.12346 & 0.07128 & 0.97627 \\
          & 500  & -0.00526 & -0.01615 & 0.02025 & 0.09437 & 0.00234 & 0.04340 & 0.12372 & 0.05468 & 0.93840 \\
          & 1000 & -0.02503 & -0.01305 & 0.00533 & 0.07822 & 0.00166 & 0.03628 & 0.11286 & 0.04548 & 0.84571 \\\hline
MSADE     & 25   & -0.08407 & -0.06121 & 0.12746 & 0.07625 & 0.03648 & 0.08094 & 0.09058 & 0.22100 & 1.08835 \\
          & 50   & -0.01539 & -0.05636 & 0.10143 & 0.06428 & 0.01849 & 0.06218 & 0.07969 & 0.15861 & 0.95779 \\
          & 100  & 0.03007  & -0.04214 & 0.09582 & 0.06139 & 0.01142 & 0.05905 & 0.07955 & 0.12180 & 0.91858 \\
          & 200  & 0.03062  & -0.02972 & 0.06332 & 0.05250 & 0.00644 & 0.04531 & 0.07155 & 0.09000 & 0.77893 \\
          & 500  & 0.03328  & -0.02170 & 0.05082 & 0.05253 & 0.00325 & 0.03299 & 0.07157 & 0.06154 & 0.65174 \\
          & 1000 & 0.03403  & -0.01784 & 0.04441 & 0.05071 & 0.00229 & 0.02990 & 0.07191 & 0.05005 & 0.61107 \\\hline
MSALDE    & 25   & -0.08659 & -0.10693 & 0.20344 & 0.16854 & 0.03492 & 0.15110 & 0.14541 & 0.22094 & 1.56349 \\
          & 50   & 0.04033  & -0.08917 & 0.18661 & 0.15495 & 0.02268 & 0.13326 & 0.14453 & 0.17745 & 1.50231 \\
          & 100  & 0.08385  & -0.07079 & 0.17507 & 0.16017 & 0.01458 & 0.12504 & 0.14971 & 0.13870 & 1.47561 \\
          & 200  & 0.07401  & -0.05513 & 0.13266 & 0.15415 & 0.00972 & 0.10843 & 0.15094 & 0.11001 & 1.37320 \\
          & 500  & 0.05240  & -0.03793 & 0.08948 & 0.12192 & 0.00547 & 0.07653 & 0.13796 & 0.07936 & 1.15536 \\
          & 1000 & 0.04050  & -0.02939 & 0.06814 & 0.09951 & 0.00386 & 0.06088 & 0.12382 & 0.06354 & 1.00230

\\\hline
\end{tabular}
}
\end{table}

From the simulation results, we observe that the bias and MSE values of the estimator for all parameters decrease as the sample sizes increase. In small sample sizes, MPSE, MSALDE and MSADE have been good competitors to MLE of $\omega$. It has been observed that the MSE values of these estimators are sometimes smaller than the MLE. For the $\kappa$ parameter, the MLE generally outperforms its competitors according to the MSE criterion. However, MLE did not outperform its competitors in estimating the p parameter as it did in estimating the $\kappa$ parameter. It can be concluded that MSADE is the best estimator overall in estimating the p parameter. In conclusion, we recommend MLE, MSADE and MSALDE for estimating the parameters of the RBTC distribution. 

\section{Real data analysis}

In the real data analysis section, we present practical data examples to demonstrate the superiority of the RBTC distribution over its competitors such as Chen (C) \cite{chen2000new}, transmuted Weibull (TW) \cite{aryal2011transmuted}, transmuted generalized Rayleigh (TGR) \cite{merovci2014transmuted}, generalized Rayleigh (GR) \cite{kundu2005generalized}, transmuted record type Weibull (TRTW) \cite{balakrishnan2021record}, transmuted exponentiated exponential (TEE) \cite{merovci2013transmuted} in modeling real-life data. For the comparison the fitted models, we consider some selection criteria as follows: Kolmogorov-Smirnov statistics (KS), Anderson-Darling statistics (AD), Cramer-von-Mises statistics (CvM), and their p-values.
\subsection{Failure time data}
In this subsection, a data example is presented to demonstrate the superiority of the RBTC distribution over its competitors in modeling real-world data. 
This censored tri-modal data studied by \cite{murthy2004weibull}, includes 30 items that are tested when testing is stopped on the 20th failure. The failure time data set is
0.0014, 0.0623, 1.3826, 2.0130, 2.5274, 2.8221, 3.1544, 4.9835, 5.5462, 5.8196, 5.8714,
7.4710, 7.5080, 7.6667, 8.6122, 9.0442, 9.1153, 9.6477, 10.1547, 10.7582. Table \ref{tab:real1} shows the selection criteria as well as the MLEs and the corresponding standard errors (SEs) of the parameters for the fitted models. Figure \ref{fig:real1cdf} illustrates the fitted CDFs while Figure \ref{fig:real1pdf} shows the fitted PDFs for failure time data.
\begin{table}[H]
\centering
\caption{The comparison statistics, MLEs and SEs of the parameters for the failure time data}%
\label{tab:real1}
\scalebox{0.65} {\ 
\begin{tabular}{cccccccccccccc}
\hline
Model & $-2\log \ell$ & KS     & AD     & CvM    & p-value(KS) & p-value(AD) & p-value(CvM) & $\hat{\omega}$    & $\hat{\kappa}$  & $\hat{p}$       & SE($\hat{\omega}$) & SE($\hat{\kappa}$) & SE($\hat{p}$)  \\\hline
RBTC                      & 95.8601  & 0.1477 & 0.6191 & 0.0642 & 0.7217      & 0.6280      & 0.7933       & 0.0837 & 0.5628 & 0.4424  & 0.0681  & 0.0817    & 0.4299 \\
C                         & 96.2592  & 0.1520 & 0.6493 & 0.0681 & 0.6897      & 0.6005      & 0.7688       & 0.0480 & 0.5887 & -       & 0.0260  & 0.0652    & -      \\
TW                        & 107.5636 & 0.1931 & 1.3186 & 0.1653 & 0.3947      & 0.2262      & 0.3483       & 0.9822 & 4.4411 & -0.5941 & 0.2119  & 1.1028    & 0.3028 \\
TGR                       & 100.3834 & 0.1953 & 1.1070 & 0.1792 & 0.3810      & 0.3050      & 0.3134       & 0.3132 & 0.1209 & -0.6711 & 0.0970  & 0.0215    & 0.2777 \\
GR                        & 103.3645 & 0.2379 & 1.6268 & 0.2901 & 0.1764      & 0.1493      & 0.1439       & 0.3945 & 0.1101 & -       & 0.1001  & 0.0216    & -      \\
TRTW                      & 106.8758 & 0.1922 & 1.2945 & 0.1635 & 0.4000      & 0.2339      & 0.3530       & 0.9667 & 0.3204 & 0.7132  & 0.2094  & 0.1580    & 0.2177 \\
W                         & 109.5036 & 0.2205 & 1.5974 & 0.2271 & 0.2465      & 0.1551      & 0.2212       & 1.0893 & 5.8164 & -       & 0.2210  & 1.2216    & -      \\
TEE                       & 106.1027 & 0.2173 & 1.4041 & 0.2263 & 0.2611      & 0.2011      & 0.2224       & 0.6908 & 0.1881 & -0.6894 & 0.2251  & 0.0511    & 0.2606    \\\hline 
\end{tabular}
}
\end{table}

\begin{figure}[H]
\centerline{\includegraphics[width=4.96in,height=2.82in]{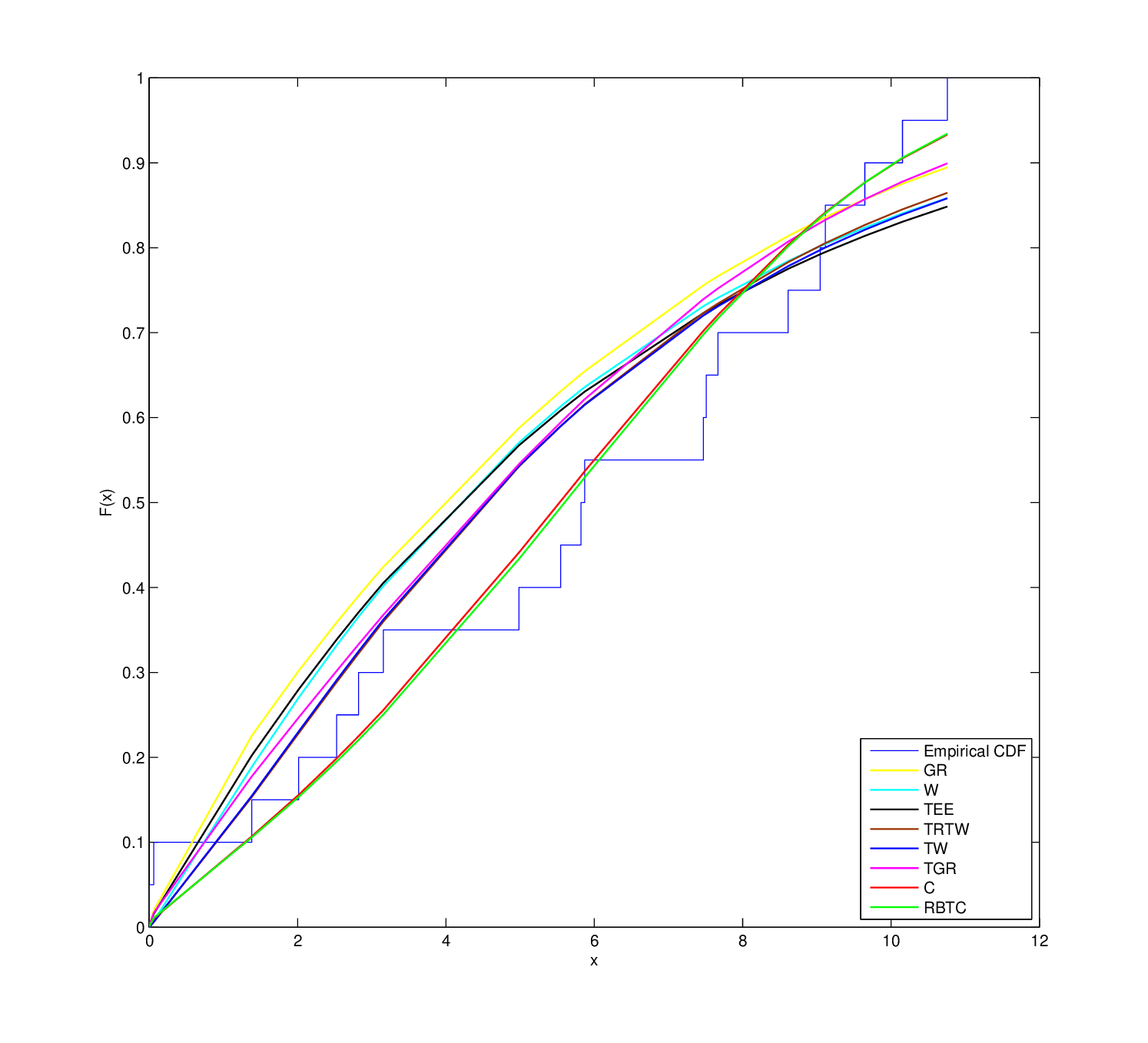}}
\caption{The fitted CDFs for the failure time data set}
\label{fig:real1cdf}
\end{figure}

\begin{figure}[H]
\centerline{\includegraphics[width=4.89in,height=2.78in]{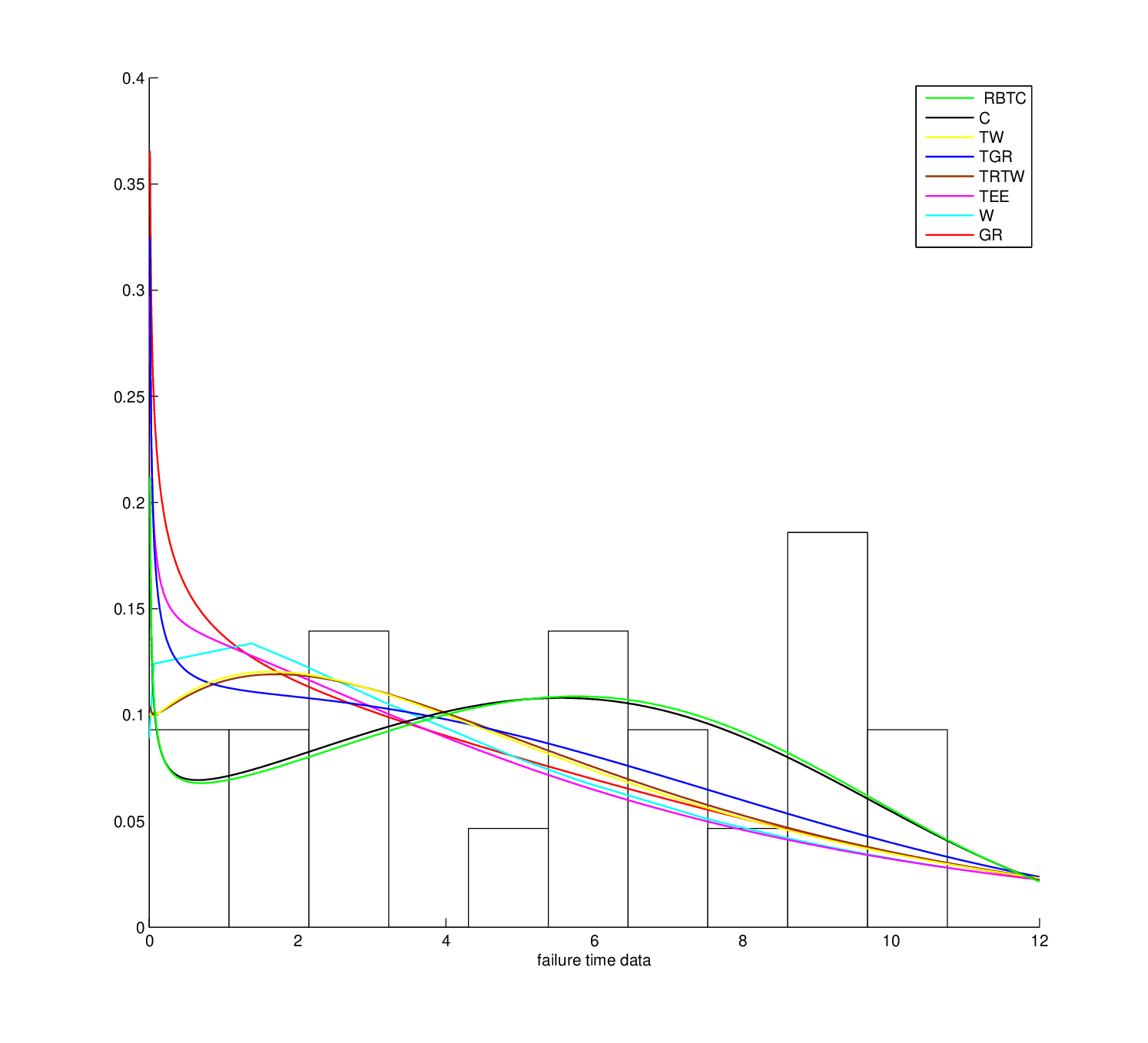}}
\caption{The fitted PDFs for the failure time data set}
\label{fig:real1pdf}
\end{figure}

From Table \ref{tab:real1}, it is clearly seen that the best-fitted distribution is the RBTC distribution for the failure time data in terms of all selection criteria.
\subsection{Iron sheet data}
In this subsection, we deal with the modelling the real-life data for RBTC distribution and its competitor ones. In this regard, we use the dataset provided by \cite{dasgupta2011distribution}, consists of 50
observations denotes to holes operation on jobs made of iron sheet. The iron sheet dataset is 0.04, 0.02, 0.06, 0.12, 0.14, 0.08, 0.22, 0.12, 0.08, 0.26, 0.24, 0.04, 0.14, 0.16, 0.08, 0.26,
0.32, 0.28, 0.14, 0.16, 0.24, 0.22, 0.12, 0.18, 0.24, 0.32, 0.16, 0.14, 0.08, 0.16, 0.24, 0.16,
0.32, 0.18, 0.24, 0.22, 0.16, 0.12, 0.24, 0.06, 0.02, 0.18, 0.22, 0.14, 0.06, 0.04, 0.14, 0.26,
0.18, 0.16. Table \ref{tab:real2} shows that the comparison statistics as weel as the MLEs and corresponding SEs of the parameters for the fitted distributions. Figure \ref{fig:real2cdf} illustrates the fitted CDFs while Figure \ref{fig:real2pdf} shows the fitted PDFs for iron sheet data.
\begin{table}[H]
\centering
\caption{The comparison statistics, MLEs and SEs of the parameters for the iron sheet data}%
\label{tab:real2}
\scalebox{0.65} {\ 
\begin{tabular}{cccccccccccccc}
\hline
Model & $-2\log \ell$ & KS     & AD     & CvM    & p-value(KS) & p-value(AD) & p-value(CvM) & $\hat{\omega}$    & $\hat{\kappa}$  & $\hat{p}$       & SE($\hat{\omega}$) & SE($\hat{\kappa}$) & SE($\hat{p}$)  \\\hline
RBTC                      & -112.9064 & 0.1039 & 0.5971 & 0.0909 & 0.6531      & 0.6499      & 0.6334       & 36.2939 & 1.8817  & 0.4872  & 15.6068 & 0.3518    & 0.3435 \\
C                         & -112.1180 & 0.1101 & 0.6792 & 0.1042 & 0.5796      & 0.5756      & 0.5663       & 33.2842 & 2.0799  & -       & 13.9247 & 0.2539    & -      \\
TW                        & -112.1179 & 0.1070 & 0.6610 & 0.0976 & 0.6163      & 0.5915      & 0.5984       & 1.9918  & 0.1708  & -0.2720 & 0.3310  & 0.0236    & 0.4307 \\
TGR                       & -112.4676 & 0.1073 & 0.5882 & 0.0957 & 0.6121      & 0.6584      & 0.6085       & 0.8452  & 5.7644  & -0.4389 & 0.2662  & 0.5097    & 0.4262 \\
GR                        & -111.5415 & 0.1263 & 0.7151 & 0.1225 & 0.4025      & 0.5456      & 0.4867       & 1.0056  & 5.5085  & -       & 0.1835  & 0.5013    & -      \\
TRTW                      & -112.5542 & 0.1057 & 0.6294 & 0.0938 & 0.6318      & 0.6198      & 0.6180       & 1.9363  & 40.4793 & 0.4758  & 0.3315  & 15.9428   & 0.3364 \\
W                         & -111.7836 & 0.1099 & 0.7084 & 0.1066 & 0.5816      & 0.5511      & 0.5553       & 2.1195  & 0.1837  & -       & 0.2463  & 0.0128    & -      \\
TEE                       & -106.5607 & 0.1462 & 1.0248 & 0.1678 & 0.2357      & 0.3442      & 0.3405       & 2.6940  & 12.4953 & -0.5471 & 0.8399  & 1.6662    & 0.3189    \\\hline 
\end{tabular}
}
\end{table}

\begin{figure}[H]
\centerline{\includegraphics[width=4.96in,height=2.82in]{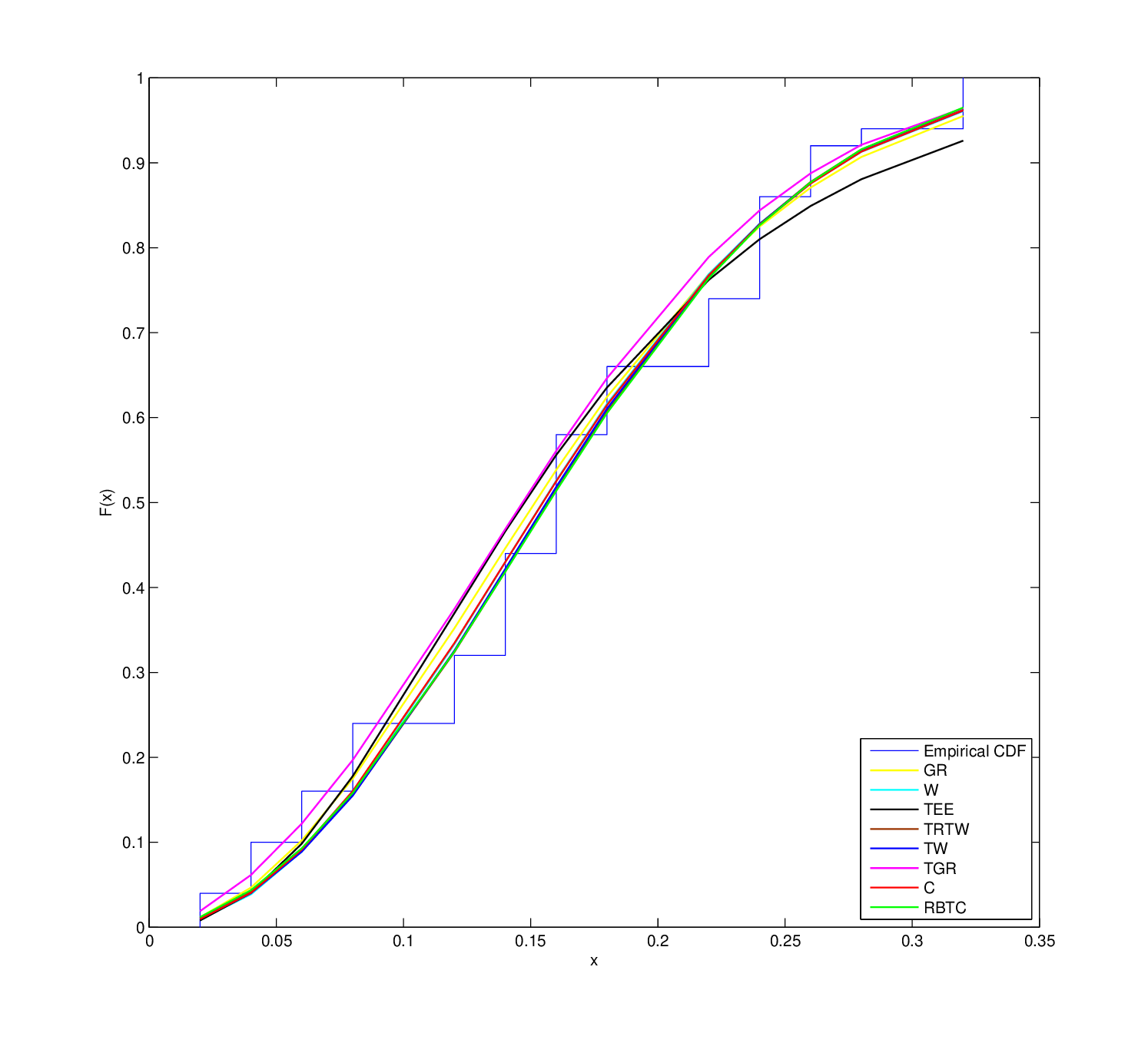}}
\caption{The fitted CDFs for the iron sheet data set}
\label{fig:real2cdf}
\end{figure}

\begin{figure}[H]
\centerline{\includegraphics[width=4.89in,height=2.78in]{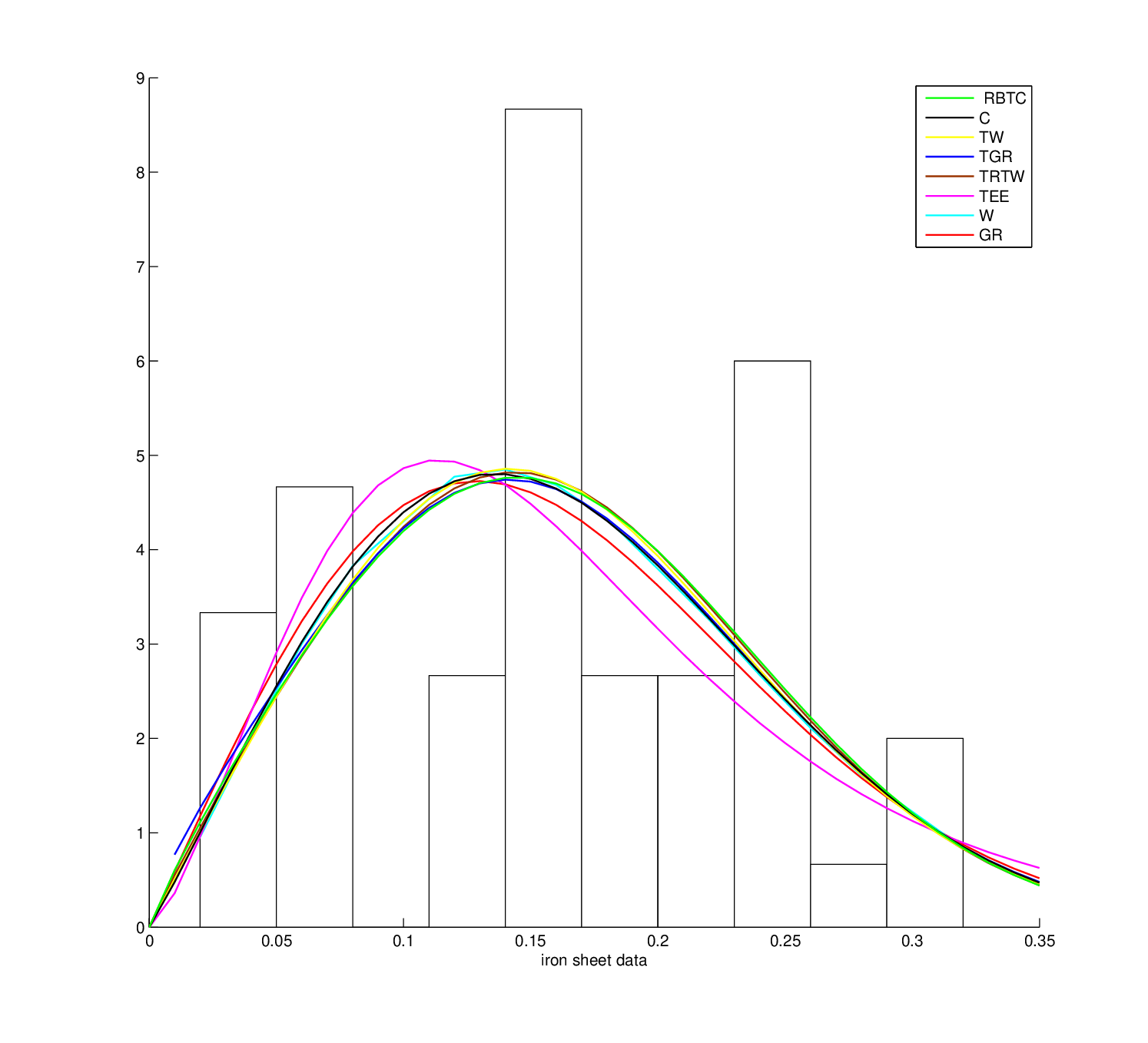}}
\caption{The fitted PDFs for the iron sheet data set}
\label{fig:real2pdf}
\end{figure}

According to Table \ref{tab:real2}, we observe that the best-fitted distribution is the RBTC distribution for the failure time data in terms of all selection criteria. Thus, we conclude that the RBTC distribution is superior to its competitors in modeling failure time and iron sheet data sets (Tables (\ref{tab:real1}-\ref{tab:real2})).

\section{Concluding Remarks}
In conclusion, we propose a new flexible alternative to Chen distribution and transmuted distributions by using the RBTM \cite{balakrishnan2021record}. The introduced distribution is called the RBTC distribution. The RBTC distribution is flexible in data modeling due to having to different hazard shapes such as increasing, decreasing and bathtub curve. We consider nine estimators for the point estimation. We use an AR algorithm to generate random sample from the RBTC distribution. Then, we utilize this algorithm in all simulations. According to findings of the MC simulation study, we recommend MSADE, MSALDE and MLE to estimate the $\omega$, $\kappa$ and p parameters. In practical data analysis, we observe that the RBTC distribution is superior to its competitors according to all selection criteria. In future work, we will suggest a novel special case of the record-based transmuted family of distribution based on log-logistic distribution. 

\bigskip
\noindent\textbf{Funding} We have not any financial support in this paper.
\newline
\bigskip
\noindent\textbf{Data Availability} The datasets are open access are given in reference list.

\bigskip
\noindent\textbf{Conflict of interest} There is no conflict of interest.

\bigskip
\noindent\textbf{Author Contribution declaration}
C.T. contributed to all sections of the article as the lone author.

\bibliographystyle{apalike} 
\bibliography{ref}
\end{document}